\documentclass[letterpaper, 10 pt, conference]{ieeeconf}   
\IEEEoverridecommandlockouts                    
\usepackage{enumerate}

\usepackage{mathtools,amsthm}

\usepackage{graphics} 
\usepackage{epsfig}
\usepackage{times} 
\usepackage{amsmath,amssymb,amsfonts}
\usepackage{mathrsfs}
\usepackage{algorithmic}
\usepackage{graphicx}
\usepackage{textcomp}
\usepackage{varioref}
\usepackage{import}
\usepackage{dsfont}
\usepackage{framed,xcolor}
\usepackage{color}
\usepackage{comment}
\usepackage{multirow}
\usepackage{epstopdf}   
\usepackage{psfrag}
\usepackage{breqn}
\usepackage{float}
\usepackage{mathtools}
\usepackage{booktabs}
\usepackage{soul}
\usepackage{amsbsy}
\usepackage[bottom]{footmisc}
\usepackage{lineno}
\usepackage{cleveref}
\usepackage{microtype}
\usepackage{comment}
\newtheorem{theorem}{Theorem}

\newtheorem{Lemma}{Lemma}
\newtheorem{Definition}{Definition}

\newtheorem{Assumption}{Assumption}
\definecolor{arash}{rgb}{0.8,0.8,1}

\newcommand{\vect}[1]{\ensuremath{\boldsymbol{\mathrm{#1}}}}
\usepackage{graphicx}
\usepackage{tabularx,booktabs}
\usepackage{multicol}
\usepackage{multimedia}
\usepackage{fancybox}
\usepackage{soul}
\makeatletter
\newcommand{\biggg}{\bBigg@{1.6}}  

\makeatother

\newcounter{lastnote}

\title{\LARGE \bf
Functional Stability of Discounted Markov Decision Processes Using Economic MPC Dissipativity Theory}
\author{Arash Bahari Kordabad and Sebastien Gros 
\thanks{This paper was supported by the
Norwegian Research Council project ``Safe Reinforcement Learning using MPC" (SARLEM).}

\thanks{The authors are with Department of Engineering Cybernetics, Norwegian University of Science and Technology (NTNU), Trondheim, Norway. E-mail:{\tt\small\{arash.b.kordabad, sebastien.gros\}@ntnu.no}}

}
\begin{document}
\bstctlcite{IEEEexample:BSTcontrol}
\maketitle
\thispagestyle{empty}
\pagestyle{empty}
\begin{abstract} This paper discusses the functional stability of closed-loop Markov Chains under optimal policies resulting from a discounted optimality criterion, forming Markov Decision Processes (MDPs). We investigate the stability of MDPs in the sense of probability measures (densities) underlying the state distributions and extend the dissipativity theory of Economic Model Predictive Control in order to characterize the MDP stability. This theory requires a so-called storage function satisfying a dissipativity inequality. In the probability measures space and for the discounted setting, we introduce new dissipativity conditions ensuring the MDP stability. We then use finite-horizon optimal control problems in order to generate valid storage functionals. In practice, we propose to use Q-learning to compute the storage functionals.
\end{abstract}
\section{Introduction}
Markov Decision Processes (MDPs) provide a generic and standard framework for optimal stochastic control of discrete-time dynamical systems, where the stage cost and transition probability depend only on the current state and the current input of the system~\cite{puterman2014markov}. For an MDP, a policy is a mapping from the state space into the input space and determines how to select the input based on the observation of the current state. Solving an MDP refers to finding an optimal policy that typically minimizes the expected value of the discounted infinite-horizon sum of stage costs. Reinforcement Learning (RL) and Dynamic programming are two common techniques to solve MDPs~\cite{sutton2018reinforcement}.

Most of the research has been done in order to find the optimal policy or verify the optimality of a given policy. However, in general, optimality may not lead to the stability of the closed-loop Markov Chain. Stability of the Markov Chains has been extensively studied in \cite{meyn2012markov}. However, this framework provides results that are not easily related to MDPs and optimality criteria. To the best of our knowledge, there are limited results characterizing the stability of MDPs as an outcome of the interplay between its objective function and its dynamics. 

In order to characterize the closed-loop stability of MDPs, we extend the concept of stability and dissipativity developed in the context of Economic Model Predictive Control (EMPC)~\cite{MPCbook}. EMPC optimizes a sum of stage costs that is not necessarily positive definite~\cite{rawlings2009optimizing}. Dissipativity is a key concept in EMPC to argue about the asymptotic stability of the closed-loop system under the optimal policy~\cite{amrit2011economic}. This theory is based on a so-called \textit{storage function} satisfying the dissipativity inequality. The storage function can be used to convert an EMPC to a \textit{tracking} MPC having a stage cost that is lower bounded by a $\mathcal{K}_\infty$ function. Under the dissipativity condition, one can show that the tracking MPC has the same optimal policy as the EMPC. Moreover, the value function resulting from the tracking MPC can be used as a Lyapunov function to show the closed-loop stability of the system under the optimal policy. 

Dissipativity is well-known for EMPC schemes having an undiscounted cost and deterministic dynamics. In the discounted setting, finding the Lyapunov function still is challenging even for positive-definite stage costs~\cite{postoyan2016stability}. In the discounted setting, the discount factor plays a vital role in the closed-loop stability. Recently the dissipativity theory has been extended to the discounted setting with deterministic dynamics~\cite{zanon2021new}. These conditions are called \textit{Strong Discounted Strict Dissipativity (SDSD)}.

We use the generalization of the classic dissipativity theory by making an argument on the measure space underlying the MDP rather than on the state space itself. This idea was first discussed in \cite{gros2021dissipativity}, but was limited to undiscounted MDPs, where the dissipativity is fairly straightforward. In this paper, we consider MDPs with a general functional stage cost. We use the concept of $D$-stability \cite{gros2021dissipativity} and introduce generalized functional dissipativity conditions for MDPs with a discounted objective function. We label these conditions Functional Strong Discounted Strict Dissipativity (FSDSD). These conditions require the transition probability, the stage cost, and the discount factor of the MDP to satisfy certain inequalities. We show that if a given problem is FSDSD, then the $D$-stability of MDP follows.

Moreover, \cite{gros2021dissipativity} covers only the stability analysis, while we discuss it in the learning context and provide practical aspects of the method. Indeed, first, we show that an undiscounted finite-horizon Optimal Control Problem (OCP) is able to capture the optimal value functionals and policy resulting from a discounted infinite-horizon OCP. Then we use a parameterized undiscounted finite-horizon OCP to approximate the action-value functional and show that this framework yields a valid storage function that satisfies FSDSD conditions. Q-learning will be proposed as a practical way of learning the OCP parameters.
\section{Problem Setting}\label{sec:PS}
In this section, we detail Markov Decision Processes (MDPs) and formulate their representation in the state density space. We consider an MDP with the following transition probability density: 
\begin{align}\label{eq:MDP}
   \xi(\vect s_{k+1}|\vect s_k,\vect a_k)\,\,,
\end{align}
where $\vect s_k \in \mathcal{X}\subset\mathbb{R}^{n}$, $\vect a_k \in \mathcal{U}\subset\mathbb{R}^{m}$, and $\vect s_{k+1}$ are the current state, input, and subsequent state, respectively, and $k\in\mathbb{I}_{\geq 0}$ is the discrete-time index. The input $\vect a_k$ applied to the system for a given state $\vect s_k$ is selected by a deterministic policy $\vect \pi : \mathcal{X} \rightarrow \mathcal{U}$. We label $\mathcal{P}$ the set of policies such that the conditional measure \eqref{eq:MDP} is $\sigma$-finite, i.e., $\vect\pi\in \mathcal{P}$. We denote $\rho_0\in\Xi$ as the initial state $\vect s_0$ distribution, i.e $\vect s_0\sim \rho_0$, where $\Xi$ is the set of measures supported on  $\mathcal{X}$. We define probability measure sequences $\rho^{\vect\pi}_{k}\in\Xi$ generated by the closed-loop Markov Chain $\xi(\vect s^{+}|\vect s,\vect \pi(\vect s))$ with policy $\vect\pi$, as:
\begin{align}
    \rho^{\vect\pi}_{k+1}(\cdot)=\mathcal{T}_{\vect\pi} \rho^{\vect\pi}_{k}(\cdot)=\int_{\mathcal{X}} \xi(\cdot|\vect s,\vect \pi(\vect s)) \rho^{\vect\pi}_k(\mathrm{d}\vect s)\,\,,
\end{align}
where $\mathcal{T}_{\vect\pi}:\Xi\rightarrow\Xi$ is defined as the transition operator on measures and $\rho^{\vect\pi}_0=\rho_0$, $\forall \vect\pi$. In general, characterization of convergence of the state sequences $\{\vect s_{k}\}_{k=0}^\infty$ resulting from the closed-loop Markov Chain $\xi(\vect s^{+}|\vect s,\vect \pi(\vect s))$ is very difficult. To tackle this issue, in this paper, instead of working with state sequences $\vect s_k$, we propose to work with probability measure sequences $\{\rho^{\vect\pi}_{k}\}_{k=0}^\infty$, describing the probability distribution of the states $\vect s_k$ over time. The selected (possibly nonlinear) stage cost functional, denoted by ${\mathcal{L}}:\Xi\times\mathcal{U}\rightarrow\mathbb{R}$, does not have a specific structure and it will be an important point in the rest of the paper. One can select it as follows:
\begin{align}\label{eq:stage:func}
    {\mathcal{L}}[\rho^{\vect\pi}_k,\vect\pi]=\mathbb{E}_{\vect s\sim \rho^{\vect\pi}_{k}} \left[ \ell(\vect s,\vect\pi(\vect s))\right]\,\,,
\end{align}
where $\ell:\mathcal{X}\times\mathcal{U}\rightarrow\mathbb{R}$ is a stage cost function. In fact, stage cost \eqref{eq:stage:func} is a particular case of functional stage cost, where it linearly depends on the stage function. Using cost functionals ${\mathcal{L}}$ that do not necessarily take the form  \eqref{eq:stage:func} is a key in this paper to discuss the functional stability of closed-loop Markov Chains. We then denote the optimal steady-state measure by $\rho^\star$ and the corresponding stage cost by ${\mathcal{L}}_0$. Without loss of generality we can assume that  ${\mathcal{L}}_0=0$ in order to have a well-posed value functional. Clearly, if this does not hold, one can shift the stage cost to achieve ${\mathcal{L}}_0=0$.
Let us consider the following discounted infinite-horizon OCP:
\begin{subequations}\label{eq:VL}
\begin{align}
    V^\star[\rho_0]=\min_{\vect\pi}&\,\, \sum_{k=0}^{\infty}\gamma^k \mathcal{L}[\rho^{\vect\pi}_k,\vect\pi] \\
    \mathrm{s.t.}& \,\, \rho^{\vect\pi}_{k+1}=\mathcal{T}_{\vect\pi} \rho^{\vect\pi}_{k},\quad \rho^{\vect\pi}_{0}=\rho_{0}\,\,,
\end{align}
\end{subequations}
where $\gamma\in(0,1)$ is the discount factor and $V^\star:\Xi\rightarrow\mathbb{R}$ is the optimal value functional. We denote the optimal policy by $\vect\pi^\star$, solution of \eqref{eq:VL}. In the following, we make a standard assumption on the stage cost functional and $V^\star$.
\begin{Assumption}\label{assum1} 
1) We assume that $\mathcal{L}[\rho,\vect\pi]$ is bounded, $\forall \rho\in\Xi,\forall \vect \pi\in\mathcal{P}$

2) There exists a non-empty set of measures, denoted by $\Xi_0$, such that for all $\rho_0\in \Xi_0$, $V^\star[\rho^{\vect\pi^\star}_k]$ remains bounded, $\forall k$.
\end{Assumption}

The optimal action-value functional ${Q}^\star$ and advantage functional ${A}^\star$ associated to \eqref{eq:VL} are defined as follows:
\begin{subequations}
\begin{align}\label{eq:Q:F}
     {Q}^\star[\rho,\vect \pi]&:={\mathcal{L}}[\rho,\vect \pi]+\gamma {V}^\star[\mathcal{T}_{\vect\pi}\rho],\quad \forall \rho\in\Xi,\forall \vect \pi\in\mathcal{P}\,,\\
     {A}^\star[\rho,\vect \pi]&:= {Q}^\star[\rho,\vect \pi]-{V}^\star[\rho],\quad \forall \rho\in\Xi,\forall \vect \pi\in\mathcal{P}\,.
\end{align}
\end{subequations}
Then from the Bellman equation, we have:
\begin{align}\label{eq:Bell}
V^\star[\rho]=Q^\star[\rho,\vect \pi^\star]=\min_{\vect \pi} Q^\star[\rho,\vect \pi] , \quad \forall \rho\in \Xi
\end{align}
One can verify the following , $\forall \rho\in \Xi$:
\begin{align}\label{eq:Bell:A}
    0&=A^\star[\rho,\vect \pi^\star]=\min_{\vect \pi} A^\star[\rho,\vect \pi],\,\,\,
    \vect\pi^\star\in \mathrm{arg}\min_{\vect \pi} A^\star[\rho,\vect \pi]\,\,.
\end{align}
We will use these results in Section \ref{sec:FA}. The next section presents the conditions on the MDP \eqref{eq:MDP} such that the sequence of measures under optimal policy converges to the optimal steady-state measure in some sense.
\section{Stability of MDPs}\label{sec:stability}
In this section, we will detail the stability of MDPs in the sense of probability measures. We extend the dissipativity theory to propose a Lyapunov functional establishing the MDP stability in the sense of $\lim_{k\rightarrow\infty} \rho^{\vect\pi^\star}_k$. In order to discuss this limit formally, we first define the following concept.
\begin{Definition} \textbf{(Dissimilarity measure)} For any $\rho,\rho'\in\Xi$, we define $D(\rho||\rho')$ as a dissimilarity measure on measure space, that maps any two measures $\rho$ and $\rho'$ to the real non-negative numbers, and $D(\rho||\rho)=0$, $\forall \rho$.
\end{Definition}

One can show that the Kullback-Leibler divergence, the  Wasserstein metric, and the total variation distance are Dissimilarity measures. Using the Dissimilarity measure concept, we can define $D$-stability of Markov Chains \cite{gros2021dissipativity}.

\begin{Definition} \textbf{($D$-stability)} The closed-loop Markov Chain $\xi(\vect s^{+}|\vect s,\vect \pi(\vect s))$ with policy $\vect\pi$ is $D$-stable with respect to the optimal steady probability measure $\rho^\star$ and dissimilarity measure $D$ if, for any $\epsilon>0$ there exists a $\delta(\epsilon)>0$ and a $K\in\mathbb{I}_{\geq 0}$ such that $D(\rho_0\|\rho^\star)<\delta(\epsilon)$ implies $D(\rho^{\vect\pi}_k\|\rho^\star)<\epsilon$, $\forall k\geq K$. Moreover, if $\lim_{k\rightarrow\infty} D(\rho^{\vect\pi}_k\|\rho^\star)=0$ holds almost everywhere, then the closed-loop Markov Chain is $D$-asymptotically stable. 
\end{Definition}
The concept of $D$-stability provides a framework to argue about $\lim_{k\rightarrow\infty} \rho^{\vect\pi^\star}_k$ in the sense of dissimilarity measures. Next lemma connects the functional stability to existence of a Lyapunov functional $V$, satisfying proper conditions.
\begin{Lemma}\label{lemma:Lyp} The closed-loop Markov Chain $\xi(\vect s^{+}|\vect s,\vect \pi^\star(\vect s))$ is $D$-asymptotically stable with respect to the  optimal steady probability measure $\rho_{\star}$ and dissimilarity measure $D$, if there exists a Lyapunov functional $V:\Xi\rightarrow\mathbb{R}^{\geq 0}$, satisfying:
\begin{subequations}
\begin{align}
     \beta_0(D(\rho_0\|\rho_{\star}))\leq &V[\rho_0]\leq \beta_1(D(\rho_0\|\rho_{\star}))\,\,,\label{eq:Lyp1}\\
     V[\mathcal{T}_{\vect\pi^\star} \rho_0]-&V[\rho_{0}]\leq -\beta_2 (D(\rho_0\|\rho_{\star}))\,\,,\label{eq:Lyp2}
\end{align}
\end{subequations}
for some $\beta_0,\beta_1, \beta_2\in\mathcal{K}_\infty$.
\end{Lemma}
\begin{proof}
    The proof can be found in \cite{gros2021dissipativity}.
\end{proof}
In the following we will connect the Lyapunov functional in Lemma \ref{lemma:Lyp} with the value functional under some conditions. Next definition develops the SDSD conditions for undiscounted MDPs, where the stage cost is a generic functional.
\begin{Definition} \textbf{(Functional Strong Discounted Strict Dissipativity (FSDSD))}
MDP \eqref{eq:MDP} with functional stage cost $ \mathcal{L}$ and discount factor $\gamma$ is \textit{Functional Strong Discounted Strict Dissipative (FSDSD)}, If there exists a bounded \textit{``storage"} functional $\lambda$ such that $\lambda [\rho^\star]=0$, satisfying:
\begin{subequations}\label{eq:SDSD}
\begin{align}
   & \mathcal{L}[\rho,\vect \pi] -\gamma \lambda [\mathcal{T}_{\vect\pi} \rho]+\lambda[\rho] \geq \alpha (D(\rho\|\rho_{\star}))\,\,,\label{eq:SDSD1}\\
        &\mathcal{L}[\rho,\vect \pi] - \lambda [\mathcal{T}_{\vect\pi} \rho]+\lambda[\rho]+(\gamma-1)V^\star[\mathcal{T}_{\vect\pi} \rho] \geq \alpha (D(\rho\|\rho_{\star}))\,.\label{eq:SDSD2}
\end{align}
\end{subequations}
for some $\alpha(\cdot)\in \mathcal{K}_\infty$\footnote{A function $\alpha: \mathbb{R}_{\geq 0}\rightarrow\mathbb{R}_{\geq 0}$ is said to belong to class $\mathcal{K}_\infty$, if $\alpha$ is continuous, strictly increasing, unbounded and $\alpha(0)=0$.} and $\forall \rho\in\Xi,\forall\vect \pi\in\mathcal{P}$, where $\rho$ is the probability measure of state $\vect s$ and $\rho^\star$ is the optimal steady measure.
\end{Definition}
Note that condition \eqref{eq:SDSD1} corresponds to the common discounted dissipativity condition~\cite{grune2016discounted}, but generalized to a functional space~\cite{gros2021dissipativity}. Condition \eqref{eq:SDSD2} has been introduced in \cite{zanon2021new} in a non-functional form to show the stability of deterministic nonlinear systems with discounted cost.
For an undiscounted setting with $\gamma\rightarrow 1$, two conditions in \eqref{eq:SDSD} coincide, and correspond to the condition proposed in \cite{amrit2011economic}. For an FSDSD problem, we define the \textit{rotated functional stage cost} $\bar {\mathcal{L}}:\Xi\rightarrow\mathbb{R}$ as follows:
\begin{align}\label{eq:rot:L}
    \bar {\mathcal{L}}[\rho,\vect\pi]=\mathcal{L}[\rho,\vect\pi] -\gamma \lambda [\mathcal{T}_{\vect\pi} \rho]+\lambda[\rho]\,\,.
\end{align}
Then if \eqref{eq:SDSD1} holds, we have:
\begin{align}\label{eq:Lbar}
\bar {\mathcal{L}}[\rho,\vect\pi]\geq     \alpha (D(\rho\|\rho_{\star}))\,\,.
\end{align}
Indeed, condition \eqref{eq:SDSD1} allows us to convert the original general stage cost $\mathcal{L}$ to the rotated stage cost $\bar{\mathcal{L}}$. For any measure $\rho$, the rotated stage cost functional $\bar{\mathcal{L}}$ is lower bounded by a $\mathcal{K}_\infty$ function applied on the selected dissimilarity measure, even if the original stage cost $\mathcal{L}$ has not such property. Next theorem relates the optimal value functional and optimal policy resulting from $\mathcal{L}$ to the optimal value functional and optimal policy resulting from $\bar{\mathcal{L}}$.
\begin{theorem}\label{theo:1}
If MDP \eqref{eq:MDP} is FSDSD, then the following discounted OCP:
\begin{subequations}\label{eq:V:bar:def}
\begin{align}
   \bar{V}^\star[\rho_0]:=\min_{\vect\pi}&\,\, \sum_{k=0}^{\infty}\gamma^k \bar{\mathcal{L}}[\rho^{\vect\pi}_k,\vect\pi]\,\,, \\
    \mathrm{s.t.}& \,\, \rho^{\vect\pi}_{k+1}=\mathcal{T}_{\vect\pi} \rho^{\vect\pi}_{k},\quad\rho^{\vect\pi}_0=\rho_0\,\,,
\end{align}
\end{subequations}
yields the same optimal policy $\vect\pi^\star$ as \eqref{eq:VL}, $\forall \rho_0\in\Xi_0$, and:
\begin{align}\label{eq:VVbar}
    \bar{V}^\star[\rho_0]=V^\star[\rho_0]+ \lambda[\rho_{0}]\,\,.
\end{align}
\end{theorem}
\begin{proof}
For an FSDSD problem, $\bar {\mathcal{L}}$ exists. Substitution of \eqref{eq:rot:L} into the cost of \eqref{eq:VL} and using a telescopic sum argument, one observes that:
\begin{align}\label{eq:proof1:1}
    &\sum_{k=0}^{\infty}\gamma^k \mathcal{L}[\rho^{\vect\pi}_k,\vect\pi]=\sum_{k=0}^{\infty}\gamma^k \left( \bar{\mathcal{L}}[\rho^{\vect\pi}_k,\vect\pi]+\gamma \lambda [\rho^{\vect\pi}_{k+1}]-\lambda[\rho^{\vect\pi}_{k}]\right)=\nonumber\\&-\lambda[\rho_0]+\lim_{N\rightarrow\infty}\gamma^N \bar{\mathcal{L}}[\rho^{\vect\pi}_N,\vect\pi]+ \sum_{k=0}^{\infty}\gamma^k \bar{\mathcal{L}}[\rho^{\vect\pi}_k,\vect\pi]=\nonumber\\&-\lambda[\rho_0]+\sum_{k=0}^{\infty}\gamma^k \bar{\mathcal{L}}[\rho^{\vect\pi}_k,\vect\pi]\,\,.
\end{align}
Note that under assumption \ref{assum1}, all terms in \eqref{eq:proof1:1} remain bounded and 
$
    \lim_{N\rightarrow\infty}\gamma^N \bar{\mathcal{L}}[\rho^{\vect\pi}_N,\vect\pi]=0\,\,.
$
Taking $\min_{\vect\pi}$ on both sides of \eqref{eq:proof1:1} results in \eqref{eq:VVbar} and the optimal policy $\vect\pi^\star$ from \eqref{eq:VL}, minimizing the right-hand side, minimizes the left-hand side as well. 
\end{proof}
Theorem \ref{theo:1} states that for an MDP that satisfies the FSDSD conditions, we can find an equivalent OCP that yields the same optimal policy and the value functional that is shifted by $\lambda$. In the next, we assume that, for any measure $\rho$, the optimal value functional $\bar {V}^\star[\rho]$ is upper bounded by a $\mathcal{K}_\infty$ function applied on the selected dissimilarity measure. This will be useful in showing Lyapunov stability. 
\begin{Assumption}\label{lemma:2} We assume following for some $\alpha_1(\cdot)\in \mathcal{K}_\infty$:
\begin{align}\label{eq:lower:V}
    \bar {V}^\star[\rho]\leq \alpha_1(D(\rho\|\rho_{\star})), \qquad \forall \rho\in \Xi_0
\end{align}
\end{Assumption}
Next theorem states that for an FSDSD MDP, $\bar {V}^\star$, defined in \eqref{eq:V:bar:def}, is a Lyapunov functional in order to prove the $D$-stability of the closed-loop Markov Chain $\xi(\vect s^{+}|\vect s,\vect \pi^\star(\vect s))$ with respect to the optimal steady measure $\rho^\star$.
\begin{theorem}\label{theo:2}
Under assumption \ref{assum1}, if the MDP with transition probability $\xi(\vect s^{+}|\vect s,\vect a)$, stage cost $\mathcal{L}$ and discount factor $\gamma$ is functional SDSD, then $\bar {V}^\star$, defined in \eqref{eq:V:bar:def}, is a Lyapunov functional for the closed-loop Markov Chain $\xi(\vect s^{+}|\vect s,\vect \pi^\star(\vect s))$ with optimal policy $\vect\pi^\star$, solution of \eqref{eq:VL}.
\end{theorem}
\begin{proof} Condition \eqref{eq:lower:V} directly implies the upper bound of \eqref{eq:Lyp1}. Using \eqref{eq:Lbar}, we have:
 \begin{align}
    \alpha(D(\rho\|\rho_{\star})) \leq \bar {V}^\star[\rho]\,\,,
 \end{align}
which results in the lower bound of \eqref{eq:Lyp1}. From the Bellman equation for OCP \eqref{eq:VL}, we have:
\begin{align}\label{eq:bell}
    \mathcal{L}[\rho^{\vect\pi^\star}_k,\vect\pi^\star] -V^\star[\rho^{\vect\pi^\star}_{k}]+\gamma V^\star[\rho^{\vect\pi^\star}_{k+1}]=0\,\,,
\end{align} 
Rearranging \eqref{eq:SDSD2} and subtracting $V^\star[\rho^{\vect\pi^\star}_{k}]$ from both sides yields:
\begin{align}\label{eq:proof2:1}
    & V^\star[\rho^{\vect\pi^\star}_{k+1}]+ \lambda[\rho^{\vect\pi^\star}_{k+1}]-V^\star[\rho^{\vect\pi^\star}_{k}]- \lambda[\rho^{\vect\pi^\star}_{k}]\leq \nonumber\\
  &-\alpha (D(\rho^{\vect\pi^\star}_{k}\|\rho_{\star}))+\mathcal{L}[\rho^{\vect\pi^\star}_k,\vect\pi^\star] -V^\star[\rho^{\vect\pi^\star}_{k}]+\gamma V^\star[\rho^{\vect\pi^\star}_{k+1}]\nonumber\\&\stackrel{(\ref{eq:bell})}=-\alpha (D(\rho^{\vect\pi^\star}_{k}\|\rho_{\star}))\,\,,
\end{align}
where we replace $\rho$, and  $\mathcal{T}_{\vect\pi}\rho$ in \eqref{eq:SDSD2} by $\rho^{\vect\pi^\star}_{k}$, $\vect\pi^\star$ and  $\rho^{\vect\pi^\star}_{k+1}$, respectively, and we have used \eqref{eq:bell} in the last equality. Then from \eqref{eq:VVbar} and \eqref{eq:proof2:1}, we have:
\begin{align}
   &\bar{V}^\star[\rho^{\vect\pi^\star}_{k+1}]-\bar{V}^\star[\rho^{\vect\pi^\star}_{k}]=V^\star[\rho^{\vect\pi^\star}_{k+1}]+ \lambda[\rho^{\vect\pi^\star}_{k+1}]\\ &\qquad\qquad -V^\star[\rho^{\vect\pi^\star}_{k}]- \lambda[\rho^{\vect\pi^\star}_{k}]\leq -\alpha (D(\rho^{\vect\pi^\star}_{k}\|\rho_{\star}))\,\,,\nonumber
\end{align}
which concludes \eqref{eq:Lyp2}. Then $\bar{V}^\star$ satisfies the  conditions of Lemma \ref{lemma:Lyp} and the closed-loop Markov Chain $\xi(\vect s^{+}|\vect s,\vect \pi^\star(\vect s))$ is $D$-asymptotically stable with respect to the  optimal steady probability measure $\rho_{\star}$ and dissimilarity measure $D$.
\end{proof}
Theorem \ref{theo:2} states the conditions that imply a $D$-asymptotically stabilizing policy for the FSDSD MDPs. However, finding the optimal policy under dissipativity conditions of theorem \ref{theo:2} and the storage functional that satisfies \eqref{eq:SDSD} are very difficult. In the next section, we address this problem by using a parameterized finite-horizon OCP scheme.
\section{Stabilizing Functional Approximator}\label{sec:FA}
Reinforcement Learning (RL) provides powerful tools to solve MDPs in practice. For instance, Q-learning is based on capturing the optimal action-value function of a given MDP, from which an optimal policy can be extracted. In this method, a parameterized action-value function is provided, and Q-learning attempts to find the optimal parameters that result in the best estimation of the optimal action-value. Deep Neural Networks (DNNs) are a common choice to provide a generic parameterization~\cite{deepRL}. However, formal analysis of the stability properties of closed-loop systems is very challenging for DNNs based function approximators. Therefore, using a more structured approximator such as the MPC scheme can be beneficial. The idea of using the function approximator based on a finite-horizon OCP has been introduced and justified in ~\cite{gros2019data}, where an EMPC was used as an approximator for RL algorithms. In fact, it has been shown that modifying stage cost and terminal cost in a parameterized MPC can capture the optimal value functions of MDPs even if an inaccurate model is used in the MPC scheme~\cite{gros2019data}. Moreover, this approximator has great capability to satisfy system constraints and safety~\cite{zanon2020safe}. Recent researches have developed further in using such approximators in the RL context~\cite{kordabad2021mpc}.

This section extends this parameterization to the functional space, where the arguments are on the measure space underlying the MDP. We use an OCP-based approximator for the optimal action-value functional to capture valid storage functional and verify the FSDSD conditions. The next theorem expresses that an undiscounted finite-horizon OCP is able to capture the optimal value functionals and policy of \eqref{eq:VL}. Note that using the undiscounted OCP will be a key to establishing stabilizing approximator results.
\begin{theorem}\label{the:f:i}
Under assumption \ref{assum1}, there exists a terminal cost functional $\hat T:\Xi\rightarrow \mathbb{R}$ and a stage cost functional $\hat L:\Xi\times\mathcal{P}\rightarrow\mathbb{R}$ such that the following undiscounted finite-horizon OCP:
\begin{subequations}\label{eq:f:v}
\begin{align}
    \hat V^\star[\rho]=\min_{\vect\pi}& \,\,\hat V^{\vect\pi}[\rho]:=\hat T[\rho^{\vect\pi}_N]+\sum_{k=0}^{N-1} \hat{L}[\rho^{\vect\pi}_k,\vect\pi],\\
    \mathrm{s.t.}& \,\, \rho^{\vect\pi}_{k+1}=\mathcal{T}_{\vect\pi} \rho^{\vect\pi}_{k},\quad \rho^{\vect\pi}_0=\rho\,\,,
\end{align}
\end{subequations}
for all $\rho\in\Xi_0$, results in the following:
\begin{enumerate}[i)]
  \item $\hat{\vect\pi}^\star=\vect\pi^\star$, \label{eq:Pi:eq}
  \item $\hat V^\star[\rho]=V^\star[\rho]$, \label{eq:VV:eq}
  \item $\hat Q^\star[\rho,\vect\pi]=Q^\star[\rho,\vect\pi]$, \label{eq:QQ:eq}
\end{enumerate}
where $\hat{\vect\pi}^\star$ is the optimal policy resulting from \eqref{eq:f:v} and:
\begin{align}\label{eq:bell:hat}
    \hat Q^\star[\rho,\vect\pi]:=\hat{L}[\rho,\vect\pi]+\hat V^\star[\mathcal{T}_{\vect\pi} \rho]\,\,.
\end{align}
\end{theorem}
\begin{proof} We select the terminal cost functional $\hat T$ and the stage cost functional $\hat L$ as follows:
\begin{subequations}\label{eq:TL}
\begin{align}
 \hat T[\rho]&= V^\star [\rho]\,\,,\\
  \hat L[\rho,\vect\pi]&= Q^\star[\rho,\vect\pi]- V^\star[\mathcal{T}_{\vect\pi}\rho]\label{eq:L:hat}\,\,.
\end{align}
\end{subequations}
Under assumption \ref{assum1}, the terminal cost and stage costs have finite values on $\Xi_0$. Substitution of \eqref{eq:TL} into \eqref{eq:f:v} and using telescopic sum, we have:
\begin{align}
\hat V^{\vect\pi}[\rho]&=  \hat T[\rho^{\vect\pi}_N]+
\sum_{k=0}^{N-1}\, \hat L[\rho^{\vect\pi}_k,\vect\pi]\nonumber\\&= V^\star[\rho^{\vect\pi}_N]+
\sum_{k=0}^{N-1}\,  Q^\star[\rho^{\vect\pi}_k,\vect\pi]- V^\star[\rho^{\vect\pi}_{k+1}]\nonumber\\&=Q^\star[\rho,\vect\pi]+ \sum_{k=1}^{N-1}\, Q^\star[\rho^{\vect\pi}_k,\vect\pi]-V^\star[\rho^{\vect\pi}_k]\nonumber\\&=Q^\star[\rho,\vect\pi]+  \sum_{k=1}^{N-1}\, A^\star[\rho^{\vect\pi}_k,\vect\pi]\,\,.
\end{align}
 From \eqref{eq:Bell:A}, we know that $\vect\pi^\star$ minimizes $A^\star[\rho^{\vect\pi}_k,\vect\pi]$ and $Q^\star[\rho,\vect\pi]$, hence it minimizes $\hat V^{\vect\pi}[\rho]$, i.e.,:
 \begin{align}
     \hat{\vect\pi}^\star=&\mathrm{arg}\min_{\vect\pi}\hat V^{\vect\pi}[\rho]\\=&\mathrm{arg}\min_{\vect\pi}Q^\star[\rho,\vect\pi]+  \sum_{k=1}^{N-1}\, A^\star[\rho^{\vect\pi}_k,\vect\pi]=\vect\pi^\star\nonumber
 \end{align}
and it yields (\ref{eq:Pi:eq}). Then substitution of the optimal policy $\vect\pi^\star$ in the cost function of \eqref{eq:f:v} reads:
\begin{align}\label{eq:VV:f}
\hat V^{\star}[\rho]=&\hat V^{\vect\pi^\star}[\rho]=Q^\star[\rho,\vect\pi^\star]+\sum_{k=1}^{N-1}\, \underbrace{A^\star[\rho_k^{\vect\pi^\star},\vect\pi^\star]}_{\stackrel{(\ref{eq:Bell:A})}=0}\stackrel{(\ref{eq:Bell})}=V^\star[\rho]\,,
\end{align}
which results in (\ref{eq:VV:eq}). Moreover, from \eqref{eq:bell:hat} and \eqref{eq:L:hat} we have:
\begin{align}
    \hat Q^\star[\rho,\vect\pi])&=\hat{L}[\rho^{\vect\pi},\vect\pi]+\hat V^\star[\mathcal{T}_{\vect\pi} \rho]\\&=Q^\star[\rho,\vect\pi]- V^\star[\mathcal{T}_{\vect\pi}\rho]+\hat V^\star[\mathcal{T}_{\vect\pi} \rho]\stackrel{(\ref{eq:VV:f})}=Q^\star[\rho,\vect\pi]\nonumber
\end{align}
which yields (\ref{eq:QQ:eq}).
\end{proof}
Theorem \ref{the:f:i} states that an undiscounted finite-horizon OCP can estimate value functional, action-value functional and optimal policy of a discounted infinite horizon OCP. Then this allows us to use a function approximator based on the undiscounted finite-horizon OCP for the discounted MDP. Similar results can be found in \cite{zanon2021stability} for value functions of classic MDPs. More specifically, let us consider the following finite-horizon undiscounted OCP as an approximator for the value functional, parameterized by $\vect\theta$:
\begin{subequations}\label{eq:MPC}
\begin{align}
   {V}_{\vect\theta}[\rho_0]=\min_{\vect\pi}&\,\, -\lambda_{\vect\theta}[\rho_0]+T_{\vect\theta}[\rho^{\vect\pi}_{N}]+\sum_{k=0}^{N-1} {\mathcal{L}_{\vect\theta}}[\rho^{\vect\pi}_k,\vect\pi] \\
    \mathrm{s.t.}& \,\, \rho^{\vect\pi}_{k+1}=\mathcal{T}_{\vect\pi} \rho^{\vect\pi}_{k},\quad \rho^{\vect\pi}_{0}=\rho_0\,\, ,
\end{align}
\end{subequations}
where  $V_{\vect\theta}$, $\lambda_{\vect\theta}$, $T_{\vect\theta}$ and $\mathcal{L}_{\vect\theta}$ are the parameterized value functional, storage functional, terminal cost and stage cost, respectively. Note that the term $-\lambda_{\vect\theta}[\rho_0]$ only depends on the first measure sequence and does not affect on the optimal policy resulting from \eqref{eq:MPC}. The term $-\lambda_{\vect\theta}[\rho_0]$ is added to the cost to have consistency with the EMPC context~\cite{MPCbook}. We denote the parameterized policy by $\vect\pi_{\vect\theta}$, solution of \eqref{eq:MPC}.

The parameterized action-value functional associated with \eqref{eq:MPC} is defined as follows:
\begin{align}\label{eq:Q:def}
    {Q}_{\vect\theta}[\rho,\vect \pi]:=-\lambda_{\vect\theta}[\rho]+{\mathcal{L}_{\vect\theta}}[\rho,\vect \pi]+{\Psi}_{\vect\theta}[\mathcal{T}_{\vect\pi}\rho]
\end{align}
where 
\begin{align}\label{eq:Psi}
 {\Psi}_{\vect\theta}[\rho]:=\lambda_{\vect\theta}[\rho]+V_{\vect\theta}[\rho]   
\end{align}
In fact, one can verify that the action-value functional ${Q}_{\vect\theta}[\rho,\vect \pi]$, value functional ${V}_{\vect\theta}[\rho]$ and policy  $\vect\pi_{\vect\theta}$ satisfy the fundamental Bellman equations. We next make a standard assumption on the terminal cost functional $T_{\vect\theta}$ and the parameterization of OCP \eqref{eq:MPC}.

\begin{Assumption}\label{assum:LT}We assume that the terminal cost functional $T_{\vect\theta}$ satisfies 
 $   T_{\vect\theta}[\rho] \geq 0,\, \forall \rho\in\Xi$.
\end{Assumption}
\begin{Assumption}\label{assum:rich}
We assume that the parameterization of \eqref{eq:MPC} is rich enough to capture the optimal action-value functional, i.e., there exists an optimal parameters vector $\vect\theta^\star$ such that:
\begin{subequations}\label{eq:C1}
\begin{align}\label{eq:QQ}
    {Q}_{\vect\theta^\star}[\rho,\vect \pi]&={Q}^\star[\rho,\vect \pi], \\ {V}_{\vect\theta^\star}[\rho]&={V}^\star[\rho]\,\,,\label{eq:VV}
\end{align}
\end{subequations}
\end{Assumption}
Assumption \ref{assum:rich} requires a universal approximator in the functional space. Note that this assumption may not hold in practice. In the next section, we detail Q-learning as a practical way to approach this assumption asymptotically.
\begin{theorem}\label{theo:3}
Under assumptions \ref{assum1}, \ref{assum:LT} and \ref{assum:rich}, $\lambda_{\vect\theta^\star}[\rho]$ satisfies \eqref{eq:SDSD}, if the following holds for some $\alpha_0(\cdot)\in \mathcal{K}_\infty$:
\begin{align}\label{eq:L:assum}
    {\mathcal{L}_{\vect\theta}}[\rho,\vect \pi]&\geq     \alpha_0 (D(\rho\|\rho_{\star})),\qquad \forall \rho\in\Xi,\,\forall\vect \pi\in\mathcal{P}
\end{align}
\end{theorem}
\begin{proof}
From assumption \ref{assum:rich}, we have:
\begin{align}\label{eq:sdsd:p1}
  &{\mathcal{L}}[\rho,\vect \pi]+\gamma {V}^\star[\mathcal{T}_{\vect\pi} \rho]\stackrel{(\ref{eq:Q:F})}={Q}^\star[\rho,\vect \pi]\stackrel{(\ref{eq:QQ})}={Q}_{\vect\theta^\star}[\rho,\vect \pi]\\&\quad\stackrel{(\ref{eq:Q:def})}=-\lambda_{\vect\theta^\star}[\rho]+{\mathcal{L}_{\vect\theta^\star}}[\rho,\vect \pi]+{\Psi}_{\vect\theta^\star}[\mathcal{T}_{\vect\pi} \rho]\nonumber\\&\quad\stackrel{(\ref{eq:Psi})}=-\lambda_{\vect\theta^\star}[\rho]+{\mathcal{L}_{\vect\theta^\star}}[\rho,\vect \pi]+\lambda_{\vect\theta^\star}[\mathcal{T}_{\vect\pi} \rho]+{V}_{\vect\theta^\star}[\mathcal{T}_{\vect\pi} \rho]\nonumber\\&\quad\stackrel{(\ref{eq:L:assum})}\geq -\lambda_{\vect\theta^\star}[\rho]+\alpha_0 (D(\rho_{0}\|\rho_{\star}))+\lambda_{\vect\theta^\star}[\mathcal{T}_{\vect\pi} \rho]+{V}^{\star}[\mathcal{T}_{\vect\pi} \rho].\nonumber
\end{align}
Rearranging \eqref{eq:sdsd:p1} results in \eqref{eq:SDSD2}. Moreover, from \eqref{eq:Psi} we can write ${\Psi}_{\vect\theta}$ as the following OCP:
\begin{subequations}\label{eq:MPC2}
\begin{align}
   {\Psi}_{\vect\theta}[\rho_0]=\min_{\vect\pi}&\,\, T_{\vect\theta}[\rho^{\vect\pi}_{N}]+\sum_{k=0}^{N-1} {\mathcal{L}_{\vect\theta}}[\rho^{\vect\pi}_k,\vect\pi] \\
    \mathrm{s.t.}& \,\, \rho^{\vect\pi}_{k+1}=\mathcal{T}_{\vect\pi} \rho^{\vect\pi}_{k},\qquad \rho^{\vect\pi}_{0}=\rho_{0}\,\,,
\end{align}
\end{subequations}
Then using \eqref{eq:L:assum} and assumption \ref{assum:LT}, the cost of \eqref{eq:MPC2} is non-negative and we have  $0 \leq {\Psi}_{\vect\theta}[\rho]$, $\forall \rho$. Then:
\begin{align}\label{eq:Psi:pos}
    0 \leq {\Psi}_{\vect\theta^\star}[\mathcal{T}_{\vect\pi} \rho]&\stackrel{(\ref{eq:Psi})}= \lambda_{\vect\theta^\star}[\mathcal{T}_{\vect\pi} \rho]+{V}_{\vect\theta^\star}[\mathcal{T}_{\vect\pi} \rho]\nonumber\\&\stackrel{(\ref{eq:VV})}=\lambda_{\vect\theta^\star}[\mathcal{T}_{\vect\pi} \rho]+{V}^{\star}[\mathcal{T}_{\vect\pi} \rho]\,\,.
\end{align}
By rearranging and multiplying both sides of \eqref{eq:Psi:pos} by the positive factor $1-\gamma$:
\begin{align}\label{eq:gamma:lam}
    &-(1-\gamma)V^\star[\mathcal{T}_{\vect\pi} \rho]\leq(1-\gamma)\lambda_{\vect\theta^\star}[\mathcal{T}_{\vect\pi} \rho]\,\,,
    \end{align}
    or equivalently
\begin{align}\label{eq:gama:lam2}
&(\gamma-1)V^\star[\mathcal{T}_{\vect\pi} \rho]-\lambda_{\vect\theta^\star}[\mathcal{T}_{\vect\pi} \rho]\leq-\gamma\lambda_{\vect\theta^\star}[\mathcal{T}_{\vect\pi} \rho]\,\,.
\end{align}
By adding $\mathcal{L}[\rho,\vect a]+\lambda_{\vect\theta^\star}[\rho]$ to both sides of \eqref{eq:gama:lam2}, we have:
\begin{align}
    &\mathcal{L}[\rho,\vect \pi]+\lambda_{\vect\theta^\star}[\rho]-\gamma\lambda_{\vect\theta^\star}[\mathcal{T}_{\vect\pi} \rho]\geq\nonumber\\&\mathcal{L}[\rho,\vect \pi]+\lambda_{\vect\theta^\star}[\rho]-\lambda_{\vect\theta^\star}[\mathcal{T}_{\vect\pi} \rho]+\nonumber\\&\qquad\qquad(\gamma-1) V^\star[\mathcal{T}_{\vect\pi} \rho]\stackrel{(\ref{eq:SDSD2})}\geq\alpha_0 (D(\rho\|\rho_{\star}))\,\,,
\end{align}
and it results in \eqref{eq:SDSD1}.
\end{proof}
Assumption \ref{assum:rich} is valid only for the FSDSD problems. In fact, for a non-FSDSD problem it is not possible to find $\vect\theta^\star$ that satisfies conditions of assumption \ref{assum:rich}. This approach enforces $D$-stability conditions for a given MDP, and if it is not stabilizable (non-FSDSD), then assumption \ref{assum:rich} is invalid. Therefore theorem \ref{theo:3} implicitly assumes that the given problem is FSDSD and states that using undiscounted finite-horizon approximator \eqref{eq:MPC} yields a valid storage functional that satisfies FSDSD conditions \eqref{eq:SDSD}. The stage cost condition \eqref{eq:L:assum} can be satisfied using constrained steps in the learning algorithm or providing a positive functional by construction. The details of these methods for deterministic systems can be found in \cite{Arash2021verification}. However, a detailed discussion on functional space is out of our scope. In general, finding $\vect\theta^\star$ that satisfies the conditions of assumption \ref{assum:rich} is very difficult. However, Q-learning is a practical way to fulfill assumption \ref{assum:rich}. Q-learning uses a Least-Square (LS) optimization and approach assumption \ref{assum:rich} asymptotically for a large number of data. Next section details this approach.

\section{Practical Implementation}\label{sec:PI}
In this section, we focus on the classic MDPs with stage cost in the form of \eqref{eq:stage:func} and a given deterministic initial state, i.e., $\rho_0=\delta_{\vect s_0}(\cdot)$, where $\delta_{\vect s_0}(\cdot)$ is the Dirac measure centered on the fixed point $\vect s_0$. This assumption is appropriate for fully observable MDPs since the current state is deterministic and available. We ought to stress here that we are still using functional stage cost functional as an important concept in the current work. Partially Observable MDPs (POMDPs) are the class of MDPs that the current state is estimated based on historical data of the system. Recently, Moving Horizon Estimation has been used in order to tackle POMDPs in combination with RL and MPC~\cite{Hossein2021MHE}. Note that all the results in the previous sections are valid, when the current state distribution is a Dirac measure. Using $\rho_0=\delta_{\vect s_0}(\cdot)$, \eqref{eq:VL} reads:

\begin{align}\label{eq:VL1}
    v^\star(\vect s_0):=&V^\star[\delta_{\vect s_0}(\cdot)]=\min_{\vect\pi}\,\, \sum_{k=0}^{\infty}\gamma^k \mathbb{E}_{\vect s\sim \rho^{\vect\pi}_{k}} \left[ \ell(\vect s,\vect\pi(\vect s))\right]\nonumber \\
    &\qquad\qquad\quad\mathrm{s.t.}\,\, \rho^{\vect\pi}_{k+1}=\mathcal{T}_{\vect\pi} \rho^{\vect\pi}_{k},\,\, \rho_0=\delta_{\vect s_0}(\cdot)
\end{align}
where $v^\star:\mathcal{X}\rightarrow\mathbb{R}$ is the classic optimal value function. The classic Bellman equation reads:
\begin{align}\label{eq:bell1}
    v^\star(\vect s_0)=\min_{\vect\pi} \ell(\vect s_0,\vect\pi(\vect s_0))+\gamma \mathbb{E}_{\vect s_1\sim \rho^{\vect\pi}_1}\left[v^\star(\vect s_1)\right]
\end{align}
where $\rho^{\vect\pi}_1=\xi(\cdot|\vect s_0,\vect \pi(\vect s_0))$. Moreover, from the Bellman equation associate to \eqref{eq:VL1}, we have:
\begin{align}\label{eq:bell2}
    v^\star(\vect s_0)=\min_{\vect\pi} \ell(\vect s_0,\vect\pi(\vect s_0))+\gamma V^\star[\rho^{\vect\pi}_1]
\end{align}
Comparing \eqref{eq:bell1} and \eqref{eq:bell2}, we have the following relation between the classic optimal value function $v^\star(\vect s)$ and the optimal value functional $V^\star[\rho]$:
\begin{align}\label{eq:VV:rel}
    V^\star[\rho^{\vect\pi}_1]=\mathbb{E}_{\vect s_1\sim\rho^{\vect\pi}_1}\left[v^\star(\vect s_1)\right]
\end{align}
The classic optimal action-value function $q^\star:\mathcal{X}\times\mathcal{U}\rightarrow\mathbb{R}$ can be defined as follows:
\begin{align}\label{eq:bell:cla}
    q^\star(\vect s,\vect \pi(\vect s)):=\ell(\vect s,\vect \pi(\vect s))+\gamma \mathbb{E}_{\vect s^+\sim \rho_1}\left[v^\star(\vect s^+)\right]
\end{align}
where $\rho_1=\xi(\cdot|\vect s,\vect \pi(\vect s))$. Substituting $\rho(\cdot)=\delta_s(\cdot)$ in  \eqref{eq:Q:F}, we have:
\begin{align}\label{eq:QstarQ}
    &{Q}^\star[\delta_{\vect s}(\cdot),\vect \pi]=\ell(\vect s,\vect \pi(\vect s))+\gamma V^\star[\xi(\cdot|\vect s,\vect \pi(\vect s))]\\&\stackrel{(\ref{eq:VV:rel})}=\ell(\vect s,\vect \pi(\vect s))+\gamma \mathbb{E}_{\vect s^+\sim \xi(\cdot|\vect s,\vect \pi(\vect s))}\left[v^\star(\vect s^+)\right]\stackrel{(\ref{eq:bell:cla})}=q^\star(\vect s,\vect \pi(\vect s))\nonumber
\end{align}
This equation shows that the optimal action-value function of a classic MDPs $q^\star(\vect s,\vect \pi(\vect s))$ can be seen as a function action-value function ${Q}^\star[\rho,\vect \pi]$ where the argument of measure $\rho$ is a Dirac measure $\delta_{\vect s}(\cdot)$. Similarly, for the parametric action-value functional one can show that:
\begin{align}\label{eq:QQ2}
    Q_{\vect\theta^\star}[\delta_{\vect s}(\cdot),\vect \pi]=q_{\vect\theta}(\vect s,\vect \pi(\vect s));
\end{align}
where $q_{\vect\theta}(\vect s,\vect \pi(\vect s))$ is a classic parameterized action-value function. Moreover, we denote the parameterized value function by $v_{\vect\theta}$. For $\rho_0=\delta_{\vect s}(\cdot)$, assumption \ref{assum:rich} reads:
\begin{align}\label{eq:QQ0}
    {Q}_{\vect\theta^\star}[\delta_{\vect s}(\cdot),\vect \pi]={Q}^\star[\delta_{\vect s}(\cdot),\vect \pi]
\end{align}
Then using \eqref{eq:QstarQ} and \eqref{eq:QQ2}, \eqref{eq:QQ0} reads:
\begin{align}\label{eq:QQ3}
    {q}_{\vect\theta^\star}(\vect s,\vect \pi(\vect s))={q}^\star(\vect s,\vect \pi(\vect s))
\end{align}
Fortunately condition \eqref{eq:QQ3} is a well-known problem in RL context, especially in the value-based algorithms. Q-learning is a common method to approach \eqref{eq:QQ3} in practice. More specifically, Q-learning uses the following LS optimization problem:
\begin{align}\label{eq:LS}
    \min_{\vect\theta}\, \mathbb E\left[\Big({q}_{\vect\theta}(\vect s,\vect \pi(\vect s)) - q^\star(\vect s,\vect \pi(\vect s))\Big)^2\right],
\end{align}
In fact LS \eqref{eq:LS} tries to find the optimal parameters vector $\vect\theta^\star$ that has the best approximation of the exact optimal action-value function $q^\star$. A richer parameterization and a larger number of  data increase the accuracy of the method. 
\section{Conclusion}\label{sec:Conc}
This paper provided a framework to analyze the functional stability of the closed-loop Markov Chains under the optimal policy resulting from minimizing the expected value of the discounted sum of stage costs for the associated MDPs. We used the dissipativity theory of EMPC in order to characterize stability properties of discounted MDPs that require a storage functional satisfying FSDSD conditions. We showed that using a function approximator based on a finite-horizon OCP allows us to obtain valid storage functional under some conditions. We focused on the Dirac measure to use the theorems in practice and addressed the use of Q-learning as a powerful RL technique to update the parameters. Considering an inaccurate model in the function approximator and providing more theoretical tools for learning and\slash or computing of the storage functional in a numerical example can be the directions of the future works.
\bibliographystyle{IEEEtran}
\bibliography{ECC2022}

\begin{thebibliography}{10}
\providecommand{\url}[1]{#1}
\csname url@samestyle\endcsname
\providecommand{\newblock}{\relax}
\providecommand{\bibinfo}[2]{#2}
\providecommand{\BIBentrySTDinterwordspacing}{\spaceskip=0pt\relax}
\providecommand{\BIBentryALTinterwordstretchfactor}{4}
\providecommand{\BIBentryALTinterwordspacing}{\spaceskip=\fontdimen2\font plus
\BIBentryALTinterwordstretchfactor\fontdimen3\font minus
  \fontdimen4\font\relax}
\providecommand{\BIBforeignlanguage}[2]{{%
\expandafter\ifx\csname l@#1\endcsname\relax
\typeout{** WARNING: IEEEtran.bst: No hyphenation pattern has been}%
\typeout{** loaded for the language `#1'. Using the pattern for}%
\typeout{** the default language instead.}%
\else
\language=\csname l@#1\endcsname
\fi
#2}}
\providecommand{\BIBdecl}{\relax}
\BIBdecl

\bibitem{puterman2014markov}
M.~L. Puterman, \emph{Markov decision processes: discrete stochastic dynamic
  programming}.\hskip 1em plus 0.5em minus 0.4em\relax John Wiley \& Sons,
  2014.

\bibitem{sutton2018reinforcement}
R.~S. Sutton and A.~G. Barto, \emph{Reinforcement learning: An
  introduction}.\hskip 1em plus 0.5em minus 0.4em\relax MIT press, 2018.

\bibitem{meyn2012markov}
S.~P. Meyn and R.~L. Tweedie, \emph{Markov chains and stochastic
  stability}.\hskip 1em plus 0.5em minus 0.4em\relax Springer Science \&
  Business Media, 2012.

\bibitem{MPCbook}
J.~B. Rawlings, D.~Q. Mayne, and M.~Diehl, \emph{Model predictive control:
  theory, computation, and design}.\hskip 1em plus 0.5em minus 0.4em\relax Nob
  Hill Publishing Madison, WI, 2017, vol.~2.

\bibitem{rawlings2009optimizing}
J.~B. Rawlings and R.~Amrit, ``Optimizing process economic performance using
  model predictive control,'' in \emph{Nonlinear model predictive
  control}.\hskip 1em plus 0.5em minus 0.4em\relax Springer, 2009, pp.
  119--138.

\bibitem{amrit2011economic}
R.~Amrit, J.~B. Rawlings, and D.~Angeli, ``Economic optimization using model
  predictive control with a terminal cost,'' \emph{Annual Reviews in Control},
  vol.~35, no.~2, pp. 178--186, 2011.

\bibitem{postoyan2016stability}
R.~Postoyan, L.~Bu{\c{s}}oniu, D.~Ne{\v{s}}i{\'c}, and J.~Daafouz, ``Stability
  analysis of discrete-time infinite-horizon optimal control with discounted
  cost,'' \emph{IEEE Transactions on Automatic Control}, vol.~62, no.~6, pp.
  2736--2749, 2016.

\bibitem{zanon2021new}
M.~Zanon and S.~Gros, ``A new dissipativity condition for asymptotic stability
  of discounted {E}conomic {MPC},'' \emph{arXiv preprint arXiv:2106.09377},
  2021.

\bibitem{gros2021dissipativity}
S.~Gros and M.~Zanon, ``A dissipativity theory for undiscounted {M}arkov
  {D}ecision {P}rocesses,'' \emph{arXiv preprint arXiv:2104.10997}, 2021.

\bibitem{grune2016discounted}
L.~Gr{\"u}ne, C.~M. Kellett, and S.~R. Weller, ``On a discounted notion of
  strict dissipativity,'' \emph{IFAC-PapersOnLine}, vol.~49, no.~18, pp.
  247--252, 2016.

\bibitem{deepRL}
V.~Mnih, K.~Kavukcuoglu, D.~Silver, A.~A. Rusu, J.~Veness, M.~G. Bellemare,
  A.~Graves, M.~Riedmiller, A.~K. Fidjeland, G.~Ostrovski \emph{et~al.},
  ``Human-level control through deep reinforcement learning,'' \emph{nature},
  vol. 518, no. 7540, pp. 529--533, 2015.

\bibitem{gros2019data}
S.~Gros and M.~Zanon, ``Data-driven economic {NMPC} using reinforcement
  learning,'' \emph{IEEE Transactions on Automatic Control}, vol.~65, no.~2,
  pp. 636--648, 2019.

\bibitem{zanon2020safe}
M.~Zanon and S.~Gros, ``Safe reinforcement learning using robust mpc,''
  \emph{IEEE Transactions on Automatic Control}, 2020.

\bibitem{kordabad2021mpc}
A.~B. Kordabad, W.~Cai, and S.~Gros, ``{MPC}-based reinforcement learning for
  economic problems with application to battery storage,'' in \emph{2021
  European Control Conference (ECC)}.\hskip 1em plus 0.5em minus 0.4em\relax
  IEEE, 2021, pp. 2573--2578.

\bibitem{zanon2021stability}
M.~Zanon, S.~Gros, and M.~Palladino, ``Stability-constrained {M}arkov
  {D}ecision {P}rocesses using {MPC},'' \emph{arXiv preprint arXiv:2102.01383},
  2021.

\bibitem{Arash2021verification}
A.~B. Kordabad and S.~Gros, ``Verification of dissipativity and evaluation of
  storage function in {E}conomic {N}onlinear {MPC} using {Q}-learning,''
  \emph{IFAC-PapersOnLine}, vol.~54, no.~6, pp. 308--313, 2021, 7th IFAC
  Conference on Nonlinear Model Predictive Control NMPC 2021.

\bibitem{Hossein2021MHE}
H.~N. Esfahani, A.~B. Kordabad, and S.~Gros, ``Reinforcement learning based on
  mpc/mhe for unmodeled and partially observable dynamics,'' in \emph{2021
  American Control Conference (ACC)}, 2021, pp. 2121--2126.

\end{thebibliography}
\end{document}